\documentclass[journal]{IEEEtran}

\ifCLASSINFOpdf
\else
   \usepackage[dvips]{graphicx}
\fi
\usepackage{url}
\usepackage{graphicx, xcolor, amsmath}

\usepackage[latin1]{inputenc}    
\usepackage{hyperref}
\hypersetup{
    colorlinks = False,
    allbordercolors = {red},   
   }
\usepackage[T1]{fontenc}

\usepackage{amsfonts,amsmath,amsthm,amssymb,mathrsfs}       
\usepackage{mathtools}
\usepackage{xifthen}
\usepackage{cite}
\usepackage[acronym]{glossaries}

\newtheorem{theorem}{Theorem}

\newtheorem{example}{Example}

\newtheorem{standing}{Standing Assumption}

\newcommand{\R}{\mathbb R}
\newcommand{\bbS}{\mathbb S}
\newcommand{\prob}[2][]{
	\ifthenelse{
		\isempty{#1}}{\mathbb{#2}}{\mathbb{#2}\left\{#1\right\}
	}
}
\newcommand{\expected}[2]{\mathbb{E}_{#1}\left[#2\right]}
\renewcommand{\span}{{\mathrm{span}}}

\usepackage[colorinlistoftodos]{todonotes}

\makeglossaries

\newacronym{LS}{LS}{least squares}
\newacronym{SDLS}{SDLS}{semidefinite least squares}
\newacronym{SD}{SD}{semidefinite}
\newacronym{SDP}{SDP}{semidefinite program}
\newacronym{iid}{i.i.d.\@}{independent and identically distributed}
\newacronym{wrt}{w.r.t.\@}{with respect to}
\newacronym{QP}{QP}{quadratic program}
\newacronym{RKHS}{RKHS}{reproducing kernel Hilbert space}

\newglossaryentry{LMI}
{
	name={LMI},
	description={linear matrix inequality},
	first={\glsentrydesc{LMI} (\glsentrytext{LMI})},
	plural={LMIs},
	descriptionplural={linear matrix inequalities},
	firstplural={\glsentrydescplural{LMIs} (\glsentryplural{LMIs})}
}

\graphicspath{{figures/}}

\begin{document}

\title{Concentration inequalities for semidefinite least squares based on data}

\author{Filippo Fabiani and Andrea Simonetto
\thanks{This work was partially supported by the ANR-JCJC grant ANR-23-CE48-0011-01.}
\thanks{F. Fabiani is with the IMT School for Advanced Studies Lucca, Piazza San Francesco 19, 55100, Lucca, Italy. A. Simonetto is with the Unit\'e de Math\'ematiques Appliqu\'ees, ENSTA, Institut Polytechnique de Paris, 91120 Palaiseau, France  (e-mail: {\footnotesize \texttt{filippo.fabiani@imtlucca.it}}, {\footnotesize \texttt{andrea.simonetto@ensta.fr}}).}}

\markboth{}{}
\maketitle

\begin{abstract}
We study data-driven least squares (LS) problems with semidefinite (SD) constraints and derive finite-sample guarantees on the spectrum of their optimal solutions when these constraints are relaxed. In particular, we provide a high confidence bound allowing one to solve a simpler program in place of the full SDLS problem, while ensuring that the eigenvalues of the resulting solution are $\varepsilon$-close of those enforced by the SD constraints. The developed certificate, which consistently shrinks as the number of data increases, turns out to be easy-to-compute, distribution-free, and only requires independent and identically distributed samples.
Moreover, when the SDLS is used to learn an unknown quadratic function, we establish bounds on the error between a gradient descent iterate minimizing the surrogate cost obtained with no SD constraints and the true minimizer.
\end{abstract}

\IEEEpeerreviewmaketitle

\section{Introduction}
\IEEEPARstart{L}{east} squares problems are one of the workhorses of signal processing and machine learning, as well as a multitude of other domains. When the underlying problems are equipped with \gls{SD} constraints, we often use the term \gls{SDLS}, whose first occurrence traces back to the 60's~\cite{brock1968optimal}. The resulting constrained programs have a number of key applications in mechanics \cite{woodgate1998efficient,krislock2003numerical,manchester2017recursive}, finance \cite{malick2004dual,boyd2005least}, stochastic control \cite{lin2009least}, and functional estimation \cite{notarnicola2022distributed}. Besides finding widespread application, one of the main challenges in \gls{SDLS} is the presence of the \gls{SD} constraint, which increases significantly the overall computational complexity. 

In particular, throughout this letter we will be interested in \textit{data-driven}, convex \gls{SDLS} of the following general form:
\begin{equation}\label{eq.sdls}
	\begin{aligned}
		&\underset{x\in\mathcal X}{\textrm{min}}&&\frac{1}{N} \|A x - b\|^2+\rho\|x\|^2\\
		&\textrm{ s.t. }&&\Lambda(F(x)) \in [m,L],
	\end{aligned}
\end{equation}
where a dataset consisting of $N$ samples $\{z^{(i)}\}_{i=1}^N$, with $z^{(i)}\coloneqq(x^{(i)},y^{(i)})$, $x^{(i)}\in\R^n$, and possibly noisy $y^{(i)}\in\R$, populate matrix $A\in\R^{N\times n}$ and vector $b\in\R^N$. 
While the convex set $\mathcal X\subseteq\R^n$ introduce inequality and equality constraints, $F:\R^n\to\bbS^\ell$ maps $x$ into a symmetric matrix relation with spectrum $\Lambda(F(x))$ constrained within $[m, L]$. In particular, $F(\cdot)$ amounts to a \gls{LMI}, $F(x)=F_0+x_1F_1+\ldots+x_nF_n$, for given symmetric matrices $F_0, F_1, \ldots, F_n\in\bbS^\ell$. Moreover, while the first part in the cost performs mean squared error minimization, the second one denotes the regularization term with $\rho>0$, making the overall cost strongly convex with a unique optimal solution.

Note that one can readily transform~\eqref{eq.sdls} into the associated matrix version by the vectorization operator $\textrm{vec}(\cdot)$ stacking the matrix columns into a vector, i.e., $x=\textrm{vec}(X)$. We can then write
$
\|A x - b\|^2 = \|C X - B\|_F^2, 
$
where $\|\cdot\|_F$ denotes the Frobenius norm, with $b=\textrm{vec}(B)$ and, by exploiting the Kronecker product, $C = I \otimes A$.
Therefore, popular instances that can be immediately recast in the form of \eqref{eq.sdls} are at least two. The first one amounts to the \textit{projection-onto-the-cone}:
$$
\underset{X\succcurlyeq0}{\textrm{min}}~\frac{1}{2}\|X - P\|_F^2,
$$
for a given matrix $P\in\bbS^\ell$, which is typically solved via singular value decomposition and eigenvalue clipping \cite{higham1988computing}. It can be also extended to $\Lambda(X)\in[m,L]$ straightforwardly.

A less obvious instance is the \textit{\gls{SD} Procrustes problem}:
$$
\underset{X\succcurlyeq0}{\textrm{min}}~\frac{1}{2}\|TX - P\|_F^2.
$$
When $T\in\R^{\ell\times\ell}$ is full rank, the factorization $X= E^\top E$ is exact and can be used to remove the \gls{SD} constraint \cite{woodgate1998efficient}. 

In general, however, the complexity of~\eqref{eq.sdls} is dominated by the SD constraint. In this letter, we ask the question of \emph{when we can safely remove} the \gls{SD} constraint and thereby transform the \gls{SDP}~\eqref{eq.sdls} into an easier program frequently turning into a \gls{QP}, according to the shape of $\mathcal X$. This can not be done in general if the \gls{SD} constraints are added for the purpose of regularizing the problem. However, when $A$ and $b$ stack noisy samples as in \eqref{eq.sdls}, the ``true'' optimal solution $x^\star\in\mathcal X$ is the one obtained with non-noisy unbiased (possibly infinite) data samples, and $\Lambda(F(x^\star)) \in [m,L]$, then one can safely relax the requirement on the span of $F(x)$, and analyze the distribution of the eigenvalues of $F(x^\star_N)$, where the (random) quantity $x^\star_N$ solves \eqref{eq.sdls} with no \gls{SD} constraints and $N$ noisy samples. 


The question of determining in probability when certain constraints are satisfied, which are so for the true process $x^\star$, goes under the umbrella of concentration inequalities, which offer bounds on how random variables deviate from a certain value (typically, the expectation). Concentration inequalities, also known as tail bounds or finite-sample certificates \cite{krikheli2021finite}, have a long history in statistical learning---see, e.g., \cite[Ch.~1]{bach2024learning}---and they have been recently used for data-driven control and decision-making \cite{hardt2016train,tsiamis2023statistical,ye2024learning,fabiani2025finite} as well as ordinary \gls{LS} \cite{bousquet2000algorithmic}. Here, we will be interested in matrix concentration inequalities aligned to the McDiarmid's one as in \cite[\S 7]{tropp2012user}. Although other type of concentration inequalities tailored for matrices exist, such as Bernstein \cite[\S 6]{tropp2012user}\cite{oliveira2009concentration} or PAC-Bayes \cite{catoni2016pac,minsker2018sub}, none of them captures the deviation between the sampled value and the expected value of a matrix function when evaluated on independent random variables. This is a key aspect in our data-driven framework. Concentration of measures as in \cite{couillet2022random}, instead, go beyond our interest.

Our contributions can then be summarized as follows:
\begin{enumerate}
\item We establish a distribution-free concentration bound on the span of $F(x^\star_N)$, where $x^\star_N$ minimizes \eqref{eq.sdls} with no \gls{SD} constraints, and it is therefore much easier to solve;
\item Based on the previous result, when the original \gls{SDLS} is designed to learn a quadratic, yet unknown, function, we show how to bound the error between a gradient descent iterate minimizing the surrogate cost obtained with no \gls{SD} constraints and the true minimizer $x^\star$.
\end{enumerate}
Finally, our theoretical results are corroborated on a numerical example addressing a quadratic function fitting problem.

%
%

\section{Examples of SDLS problems}
To fully motivate the problem addressed in this letter, we discuss next several applications of the \gls{SDLS} problem in \eqref{eq.sdls}.

\begin{example}[Fitting a quadratic function \cite{notarnicola2022distributed}]\label{sec:fitting_quadratic}
Given an unknown, quadratic, positive \gls{SD} function,
$f(x)\coloneqq x^\top Q x + c^\top x + r$, $Q\succcurlyeq 0$, $c\in\R^n$ and $r\in\R$,
the task is to estimate the latter parameters through $(\hat Q, \hat c, \hat r)$, thereby producing $\hat f(x)\coloneqq x^\top \hat Q x + \hat c^\top x + \hat r$, through measurements $y^{(i)} = f(x^{(i)}) + \eta^{(i)}$ collected at $x^{(i)}\in\R^n$ and affected by noise $\eta^{(i)}$, $i \in \{1, \ldots, N\}$.  We can formulate the problem as:
$$
\underset{\hat Q\succcurlyeq0, \hat c, \hat r}{\textnormal{\textrm{min}}}~\frac{1}{N} \sum_{i=1}^N (\hat f(x^{(i)}) - y^{(i)})^2,
$$
which can be readily transformed into:
\begin{equation}\label{eq:fitting}
	\underset{\xi \in \R^{n^2 + n + 1}}{\textnormal{\textrm{min}}}~\frac{1}{N} \|A \xi - b\|^2~\textnormal{\textrm{ s.t. }}~F(\xi) \succcurlyeq 0, 
\end{equation}
for a suitable data matrix $A$, vector $b$, and \gls{LMI} $F$. In general, we may also want to impose $\Lambda(F(\xi))\in [m,L]$, for all $\xi$. 
\hfill$\square$
\end{example}

\begin{example}[Kernel ridge regression~\cite{aubin2020hard}]\label{sec:fitting_convex}

Akin to the above, we now want to estimate some smooth, strongly convex function $f(x)$ through noisy measurements at $x^{(i)}\in\R^n$, $i \in \{1, \ldots, N\}$. By assuming that $f$ belongs to a \gls{RKHS} with kernel function $k(\cdot,\cdot)$ and bounded \gls{RKHS} norm, exploiting the reproducing property makes our goal to finding  coefficients $\alpha \in \R^N$ so that:
\begin{equation}\label{eq:krr}
	\begin{aligned}
		&\underset{\alpha \in \R^N}{\textnormal{\textrm{min}}}&&\frac{1}{N}\|\mathrm{K} \alpha - y \|^2 + \rho \|\alpha\|^2\\
		&~\textnormal{\textrm{ s.t.}}&&\Lambda([\nabla^2 \kappa(x^s)]\cdot\alpha) \in [m,L], \, s \in \mathcal{S},
	\end{aligned}
\end{equation}
where $\mathrm{K}\in\bbS^N$ is the so-called Gram matrix, with entries $\mathrm{K}_{i,j}=k(x^{(i)},x^{(j)})$, and $y=[f(x^{(1)})+\eta^{(1)} \ \cdots \ f(x^{(N)})+\eta^{(N)}]^\top$ obtained by means of noisy samples. By letting $\kappa(x)\coloneqq[k(x,x^{(1)}) \ \cdots \ k(x,x^{(N)})]^\top$, in the constraints we then indicate with the ``dot''  a weighted sum over the third tensor dimension, thereby forcing the eigenvalues of the Hessian associated to the estimate of $f(x)$ within $[m,L]$ at specific points $x^s\in\R^n$, $s \in \mathcal{S}$, to be selected.
\hfill$\square$
\end{example}

\begin{example}[Elasticity and inertia estimation \cite{woodgate1998efficient,manchester2017recursive}]
Given a matrix of data containing applied forces $F$ and one collecting the resulting measured displacements $X$, one wishes to reconstruct the elasticity matrix $K$ as follows:
$$
\underset{K \succcurlyeq 0}{\textnormal{\textrm{min}}}~\frac{1}{2}\|KX - F\|_F^2.
$$
The above can be generalized to estimate the matrix of inertia for rigid bodies, such as satellites, with angular matrix $M$ and torque matrix $T$. For a suitable matrix $R$, we then have:
\[
\underset{J \succ 0}{\textnormal{\textrm{min}}}~\frac{1}{2}\|JM - T\|_F^2~\textnormal{\textrm{ s.t. }}~RJ \geq 0. 
\tag*{$\square$}
\]
\end{example}

\begin{example}[Covariance fitting \cite{lin2009least}] Given a stochastic linear system $\dot x=A x+B d$, with $A\in\R^{n\times n}$ being a Hurwitz matrix, and $B\in\R^{n\times m}$ mapping disturbance $d\in\R^m$ over the state, the \gls{LS} covariance fitting problem can be formulates as:
$$
\begin{aligned}
	&\underset{X \succcurlyeq 0, H}{\textnormal{\textrm{min}}}&&\frac{1}{2}\|X - \Sigma\|_F^2\\
	&\textnormal{~\textrm{ s.t. }}&&AX + X A^\top = -(BH + H^\top B^\top),
\end{aligned}
$$
where $\Sigma=\tfrac{1}{N}\sum_{i=1}^N x^{(i)}{x^{(i)}}^\top\succcurlyeq 0$ is the sample covariance matrix obtained through $N$-state measurements $x^{(i)}\in\R^n$.
\hfill$\square$
\end{example}

\section{Finite-sample guarantees for \gls{SDLS}}
%
%
%
%
Next, we establish the main result of our paper, i.e., we derive finite-sample guarantees on the spectrum of $F(x)$ in \eqref{eq.sdls} when the constraints $\Lambda(F(x))\in[m,L]$ is relaxed. To this end, we will consider a simplified version of \eqref{eq.sdls} that does not include any \gls{SD}-type of constraints $\Lambda(F(x))$, namely:
\begin{equation}\label{eq.sdls2}
		\underset{x\in\mathcal X}{\textrm{min}}~\frac{1}{N} \|A x - b\|^2+\rho\|x\|^2.
\end{equation}
Let $x^\star_N$ denote the resulting optimal solution, where the subscript emphasizes its data-driven nature, i.e., $x^\star_N=x^\star_N(z^{(1)},\ldots,z^{(N)})$. Note that $x^\star_N$ is inherently random according to the \textit{unknown} distribution $\prob{P}$ underlying data. Our goal is then to establish probabilistic guarantees on how the span of $F(x^\star_N)$ concentrates as $N$ grows. To achieve our goal, we postulate the following working assumptions:
\begin{standing}\label{ass:general}
The \gls{SDLS} in \eqref{eq.sdls}, and hence \eqref{eq.sdls2}, relies on $N$ \gls{iid} samples $\{z^{(i)}\}_{i=1}^N$ drawn according to some unknown probability distribution $\mathbb{P}$, and appropriately stacked into matrix $A$ and vector $b$. The data is noisy but unbiased, and the underlying true process $x^\star$, for which we collect samples $\{z^{(i)}\}_{i=1}^N$ and that we want to estimate, verifies all the constraints in \eqref{eq.sdls}.
\hfill$\square$
\end{standing}

 Armed with these, we can therefore prove what follows:
\begin{theorem}\label{th:concentration}
	Fix $\delta\in(0,1)$ and $\rho>0$. 
	Then, there exists $\varepsilon=\varepsilon(\ell,\delta,N)>0$ such that, with probability at least $1-\delta$,
	\[
		\Lambda(F(x^\star_N))\in\left[m-\varepsilon,~L+\varepsilon\right].
		 \tag*{$\square$}
	\]
\end{theorem}
\begin{proof}
	We use $F(\cdot)$ as an application mapping the $N$ data via $x^\star_N$ into a symmetric matrix of dimension $\ell$, i.e., $F(x^\star_N)=F(x^\star_N(z^{(1)},\ldots,z^{(N)}))$, and then apply \cite[Cor.~7.5]{tropp2012user}.
	
	\sloppy To this end, we first need to prove the bounded difference property for $F(x^\star_N)-F(x^\star_{N^i})=F(x^\star_N(z^{(1)},\ldots,z^{(N)}))-F(x^\star_N(z^{(1)},\ldots,z^{(i-1)},z',z^{(i+1)},\ldots,z^{(N)}))$, where $z'$ is some sample drawn according to $\prob{P}$, \gls{iid} \gls{wrt} the dataset $\{z^{(i)}\}_{i=1}^N$, which replaces the arbitrary $i$-th one. Specifically, we have to identify matrices $\Xi_i\in\bbS^\ell$ so that $(F(x^\star_N)-F(x^\star_{N^i}))^2\preccurlyeq \Xi_i^2$, which allow us to define $\sigma^2=\|\sum_{i=1}^N \Xi_i^2\|$. By analyzing $(F(x^\star_N)-F(x^\star_{N^i}))^2$ more in detail we obtain:
	\[
	\begin{aligned}
		&(F(x^\star_N)-F(x^\star_{N^i}))^2=\left(\left(
		(x^\star_N-x^\star_{N^i})^\top\otimes I_\ell\right) \begin{bmatrix}
			F_1\\
			\vdots\\
			F_n
		\end{bmatrix}\right)^2\\
		&\stackrel{(a)}{=}\left(
		(x^\star_N\!-\!x^\star_{N^i})^\top\!\otimes\! I_\ell\right)\underbrace{\begin{bmatrix}
			F_1\\
			\vdots\\
			F_n
		\end{bmatrix}\begin{bmatrix}
			F_1~\ldots~F_n
		\end{bmatrix}}_{\eqqcolon H}\left(
		(x^\star_N\!-\!x^\star_{N^i})\!\otimes\! I_\ell\right)\\
		&\stackrel{(b)}{\preccurlyeq} \lambda_{\max}(H) \left(
		(x^\star_N-x^\star_{N^i})^\top\otimes I_\ell\right)\left(
		(x^\star_N-x^\star_{N^i})\otimes I_\ell\right)\\
		&\stackrel{(c)}{=}\lambda_{\max}(H) \|x^\star_N-x^\star_{N^i}\|^2 I_\ell,
	\end{aligned}
	\]
	where $(a)$ follows from the symmetry of $F(x)$, for all $x\in\R^n$, $(b)$ since $H\in\bbS^{\ell n}_{\succcurlyeq0}$ by construction, while $(c)$ from standard properties of the Kronecker product. To upper bound the term $\|x^\star_N-x^\star_{N^i}\|$, we rely on admissibility-type of arguments from \cite{bousquet2002stability}. In particular, by denoting with $x^\star_{N-1}$ the optimal solution to \eqref{eq.sdls2} once removed an arbitrary sample from the $N$ available $\{z^{(i)}\}_{i=1}^N$, direct application of \cite[Lemma~21]{bousquet2002stability} to $\|x^\star_N-x^\star_{N^i}\|\le\|x^\star_N-x^\star_{N-1}\|+\|x^\star_{N^i}-x^\star_{N-1}\|$ leads to $\|x^\star_N-x^\star_{N^i}\|^2\le 2B^2/\rho^2N^2$. Here, $B$ denotes some term upper bounding both $x$ and the data-based quantities $(A,b)$ that one can obtain, which happens to be finite in view of Standing Assumption~\ref{ass:general}. Specifically, it follows from i) the well-posedness of the learning problem, which rules out the possibility to collect possibly noisy samples that are unbounded, and ii) the fact that \eqref{eq.sdls2} contains a regularization term singling out a solution even if the underlying feasible set $\mathcal X$ may be unbounded.  Thus, we obtain $\sigma=\sqrt{\|\sum_{i=1}^N\tfrac{2B^2\lambda_{\max}(H)}{\rho^2N^2}I_\ell\|}=\tfrac{B}{\rho}\sqrt{\tfrac{2\lambda_{\max}(H)}{N}}$.
	
	Then, \cite[Cor.~7.5]{tropp2012user} establishes the following relation:
	\begin{equation}\label{eq.tropp_concentration}
		\prob[\lambda_{\max }(F(x^\star_N)-\expected{\prob{P}}{F(x^\star_N)})\le\epsilon]{P^\mathnormal{N}}\geq 1-\ell e^{-\epsilon^2 / 8 \sigma^2}.
	\end{equation}
	Since we have assumed unbiasedness, from the above we directly obtain: $x^\star=\expected{\prob{P}}{x^\star_N}\implies F(x^\star)=F(\expected{\prob{P}}{x^\star_N})=\expected{\prob{P}}{F(x^\star_N)}$ in view of the linearity of $F(\cdot)$. With this,
	\begin{equation}\label{dummy}
	\begin{aligned}
		\lambda_{\max }&(F(x^\star_N)-\expected{\prob{P}}{F(x^\star_N)})\le\epsilon\\
		&\iff\lambda_{\max }(F(x^\star_N)-F(x^\star))\le\epsilon\\
		&\iff F(x^\star_N)-F(x^\star)\preccurlyeq\epsilon I_\ell \\
		&\iff F(x^\star_N)\preccurlyeq\epsilon I_\ell+F(x^\star)\\
		&\implies F(x^\star_N)\preccurlyeq(\epsilon+L) I_\ell\implies\lambda_{\max }(F(x^\star_N))\le\epsilon+L.
	\end{aligned}
	\end{equation}
	Then, by deriving $\epsilon$ from the relation $\delta=\ell e^{-\epsilon^2 / 8 \sigma^2}$, we end up with the following expression:
	\[
	\begin{aligned}
		&\prob[\lambda_{\max }(F(x^\star_N))\le L+2\sigma\sqrt{2\ln\left(\ell/\delta\right)}]{P^\mathnormal{N}}\geq 1-\delta, \text{ i.e.,}\\
		&\prob[\lambda_{\max }(F(x^\star_N))\le L\!+\!\tfrac{4B}{\rho\sqrt{N}}\sqrt{\lambda_{\max}(H)\ln\left(\ell/\delta\right)}]{P^\mathnormal{N}}\geq 1\!-\!\delta,
	\end{aligned}
	\]
	which holds true for arbitrarily high confidence $\delta\in(0,1)$.

	
	To derive the relation for the ``opposite direction'' involving $\lambda_{\textrm{min}}(F(x^\star_N))$, from the proof of \cite[Cor.~7.5]{tropp2012user} one is able to obtain also the following concentration bound:
	\[
	\prob[\lambda_{\max }(\expected{\prob{P}}{F(x^\star_N)}-F(x^\star_N))\le\epsilon]{P^\mathnormal{N}}\geq 1-\ell e^{-\epsilon^2 / 8 \sigma^2}.
	\]
Then, we proceed as in~\eqref{dummy}, with the opposite signs to arrive at:	
	\[
	\prob[\lambda_{\textrm{min}}(F(x^\star_N))\ge m-\epsilon]{P^\mathnormal{N}}\geq 1-\ell e^{-\epsilon^2 / 8 \sigma^2}.
	\]
	Again, by obtaining $\epsilon$ from $\delta=\ell e^{-\epsilon^2 / 8 \sigma^2}$, and putting everything together we finally arrive at the following expression:
	\[
	\begin{aligned}
		&\mathbb P^\mathnormal{N}\left\{\Lambda(F(x^\star_N))\in\left[m-\tfrac{4B}{\rho\sqrt{N}}\sqrt{\lambda_{\max}(H)\ln\left(\ell/\delta\right)},\right.\right.\\
		&\hspace{2.8cm}\left.\left. L+\tfrac{4B}{\rho\sqrt{N}}\sqrt{\lambda_{\max}(H)\ln\left(\ell/\delta\right)}\right]\right\}\geq 1-\delta.
	\end{aligned}
	\]
	Setting $\varepsilon=\tfrac{4B}{\rho\sqrt{N}}\sqrt{\lambda_{\max}(H)\ln\left(\ell/\delta\right)}$ concludes the proof.
\end{proof}

As a crucial consequence of Theorem~\ref{th:concentration}, instead of solving the \gls{SDLS} in \eqref{eq.sdls}, one can then focus on the program in \eqref{eq.sdls2}, which is easier to solve as it frequently reduces to a \gls{QP}, according to the shape of $\mathcal X$, and still obtain finite-sample guarantees that the spectrum of $F(x^\star_N)$ will span a neighborhood of the desired one, i.e., $[m,L]$, which shrinks as $N$ grows. The established bound is easy-to-compute, distribution-free, and only requires a dataset $\{z^{(i)}\}_{i=1}^N$ with \gls{iid} samples.

\section{Optimization over learned functions}\label{sub:optimization}

Focusing on Example~\ref{sec:fitting_quadratic}, we now want to bound the distance between $x^\star=\textrm{argmin}_{x\in\R^n}~f(x)$, where $f(x)=x^\top Q x + c^\top x + r$ is the unknown function we wish to learn, and the generic iterate $x_{k+1}$ of the gradient descent scheme applied to $\hat f(x)=x^\top \hat Q^\star_N x + \hat{c}_N^{\star^\top} x + \hat r_N^\star$. The latter function is obtained, instead, by solving the variant of \eqref{eq:fitting} without the constraints $\Lambda(F(\xi)) \in [m, L]$, 
yielding  $\xi_N^\star=[\textrm{vec}(\hat Q_N^\star)^\top~\hat{c}_N^{\star^\top}~\hat r_N^\star]^\top$. 

We note that without imposing suitable constraints on the eigenvalues of the Hessian $\hat Q\succcurlyeq0$ does not ensure the convergence for the gradient descent scheme when applied to minimize $\hat f(x)$. Then, let
\begin{equation}\label{gds}
x_{k+1} = x_k - \gamma \nabla \hat f(x_k)=x_k - \gamma(\hat Q_N^\star x_k + \hat c_N^\star),
\end{equation}
denoting the generic $k$-th update of the gradient descent method, for given stepsize $\gamma>0$ and initial condition $x_0\in\R^n$. We can claim the following bound based on Theorem~\ref{th:concentration}:
\begin{theorem}\label{th:concentration_optimization}
	Fix $\delta\in(0,1)$ and $\rho>0$. Assume that $x^\star$ exists and is bounded. Then, there exists $\varepsilon=\varepsilon(\ell,\delta,N)>0$ such that, if we select a $N$ that guarantees $m- \varepsilon>0$ with probability $1-\delta$ and we select $\gamma<2/(L+\varepsilon)$, then with the same probability, the iterations~\eqref{gds} converge as follows:
	\[
	\limsup_{k\to \infty}\|x_{k+1} - x^\star\| = O(\varepsilon). \tag*{$\square$}
	\]
\end{theorem}
\begin{proof}
	By analyzing the iteration error, we readily obtain:
	\[
	\begin{aligned}
		x_{k+1} - x^\star&=x_k - \gamma(\hat Q_N^\star x_k + \hat c_N^\star)-x^\star\\
		&=x_k - \gamma(\hat Q_N^\star x_k + \hat c_N^\star)-x^\star+ \gamma(Q x^\star + c) \\
		&\hskip-1cm = (I-\gamma \hat Q_N^\star) (x_k - x^*) + \gamma (Q - \hat Q_N^\star)x^\star + \gamma (c-\hat c_N^\star).
	\end{aligned}
	\]
By choosing $\gamma$ and $N$ as specified, we impose that the matrix $(1-\gamma \hat{Q}_N^\star)$ is smaller than one in norm, with probability at least $1-\delta$. Furthermore, the error terms $Q-Q_N^\star$ and $c-\hat c_N^\star$ can be bounded following the proof of Theorem~\ref{th:concentration}. The former is immediate leading to $\|Q-\hat Q_N^\star\| \leq \varepsilon$; the latter is also easy by looking at the diagonal map $F'\coloneqq\textrm{diag}(c)$ and using the same steps to obtain $\|c-\hat c_N^\star\| \leq \varepsilon$. With this, and with probability at least $1-\delta$, 
$$
\|x_{k} - x^\star\|\le \|I-\gamma \hat Q_N^\star\|^k \|x_{0} - x^\star\| + O(\varepsilon),
$$
and the theorem is proven.
\end{proof}

The theorem ensures convergence of the sampled gradient method to an error ball of size proportional to $\varepsilon \propto \frac{1}{\sqrt{N}}\sqrt{\ln(\ell/\delta)}$. The more the points $N$, the smaller the error.  

%

\section{Numerical experiments}
We now corroborate our theoretical results on a numerical instance of Example~\ref{sec:fitting_quadratic}. The simulations are run in MATLAB on a laptop with an Apple M2 chip featuring an 8-core CPU and 16 GB RAM, while Mosek \cite{mosek} has been used as \gls{SDP} solver, implemented in YALMIP environment \cite{Lofberg2004}.

Then, we randomly generate a quadratic function $f(x)$ that we wish to learn. By denoting with $\mathcal U([a,b])$ the uniform distribution over the interval $[a,b]$, we set $Q=U^\top U\succcurlyeq0$, with each entry $U_{ij}\sim\mathcal U([0,1])$, $c\sim\mathcal U([0,1]^n)$, and $r\sim\mathcal U([0,1])$. 

\begin{figure}
	\centering
	\includegraphics[width=.8\columnwidth]{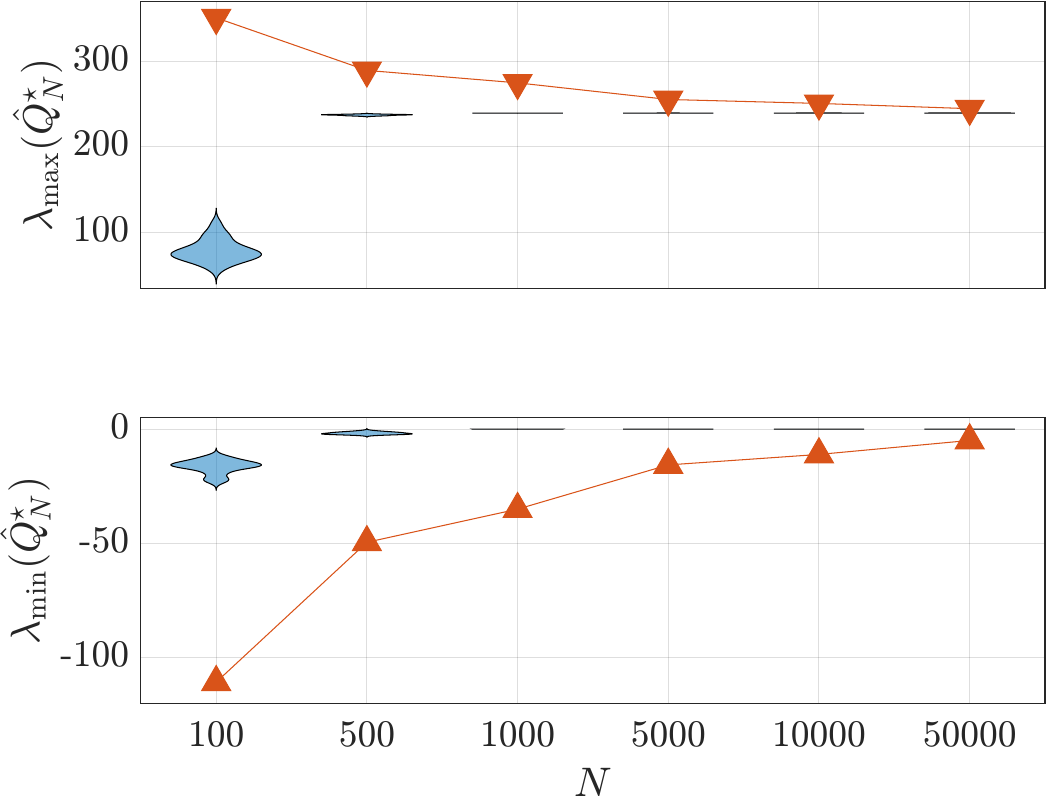}
	\caption{Violin plots reporting the maximum (top figure) and minimum (bottom figure) eigenvalues of $\hat Q^\star_N$ obtained by solving \eqref{eq:fitting}, averaged over $20$ trials with datasets of different size $N$. The red downward (respectively, upward)-pointing triangles denote the upper (resp., lower) bound in Theorem~\ref{th:concentration}.}
	\label{fig:theoretical_bound}
\end{figure}

First, we test how the bound established in Theorem~\ref{th:concentration} changes by considering a dataset of different size $N$. In particular, the latter consists of pairs $(x^{(i)},	f(x^{(i)})+\eta^{(i)})$ with each $x^{(i)}\in\mathcal U([-10,10]^n)$ and $\eta^{(i)}\sim\mathcal N(0,1)$, where $\mathcal N(\mu,\sigma^2)$ denotes the normal distribution with mean $\mu$ and variance $\sigma^2$. Matrix $A$ and vector $b$ in \eqref{eq:fitting} then reads as follows:
\[
\begin{aligned}
	A=\left[\begin{smallmatrix}
		{x^{(1)}}^\top\otimes~{x^{(1)}}^\top & & {x^{(1)}}^\top & 1\\
		&\vdots &&\\
		{x^{(N)}}^\top\otimes~{x^{(N)}}^\top && {x^{(N)}}^\top & 1
	\end{smallmatrix}\right]
	,\, 
	b=\left[\begin{smallmatrix}
		f(x^{(1)})~+~\eta^{(1)}\\
		\vdots\\
		f(x^{(N)})~+~\eta^{(N)}
	\end{smallmatrix}\right].
\end{aligned}
\]
Notice that $\xi=[\textrm{vec}(\hat Q)^\top~\hat c^\top~\hat r]^\top\in\R^{n^2+n+1}$.
Since we will solve \eqref{eq:fitting} by adding an extra regularization term $\rho\|\xi\|^2$ in the cost and imposing symmetry only, i.e., simple equality constraints $\hat Q=\hat Q^\top$ and no \gls{SD} ones, i.e., $\hat Q\succcurlyeq0$ and $\hat Q\preccurlyeq\lambda_{\max}(Q)$, we are therefore primarily interested in guaranteeing bounds  $\lambda_{\textrm{min}}(\hat Q_N^\star)\ge-\tfrac{4B}{\rho\sqrt{N}}\sqrt{n\ln\left(n/\delta\right)}$ and $\lambda_{\max}(\hat Q_N^\star)\le\lambda_{\max}(Q)+\tfrac{4B}{\rho\sqrt{N}}\sqrt{n\ln\left(n/\delta\right)}$ with probability at least $1-\delta$, since it turns out that, in the considered setting, $\ell=n$ and $\lambda_{\max}(H)=n$. To this end, Fig.~\ref{fig:theoretical_bound} allows one to contrast the theoretical bound established in Theorem~\ref{th:concentration} and the maximum/minimum eigenvalues of $\hat Q_N^\star$ obtained by solving the \gls{QP} in \eqref{eq:fitting}. In this case, the simulation is run by setting $\rho=1$, $n=30$, $\delta=0.05$, and for each $N$ averaged over $20$  different dataset $\{(x^{(i)},	f(x^{(i)})+\eta^{(i)})\}_{i=1}^N$. While one can observe a sort of monotonic behavior for the theoretical bounds, the computed $\lambda_{\max}(\hat Q_N^\star)$ and $\lambda_{\textrm{min}}(\hat Q_N^\star)$ instead do not follow a specific trend, other than producing better (i.e., close to $240.14=\lambda_{\max}(Q)$ and $0=\lambda_{\textrm{min}}(Q)$, respectively) estimates with very little variance for larger datasets.

\begin{figure}
	\centering
	\includegraphics[width=.8\columnwidth]{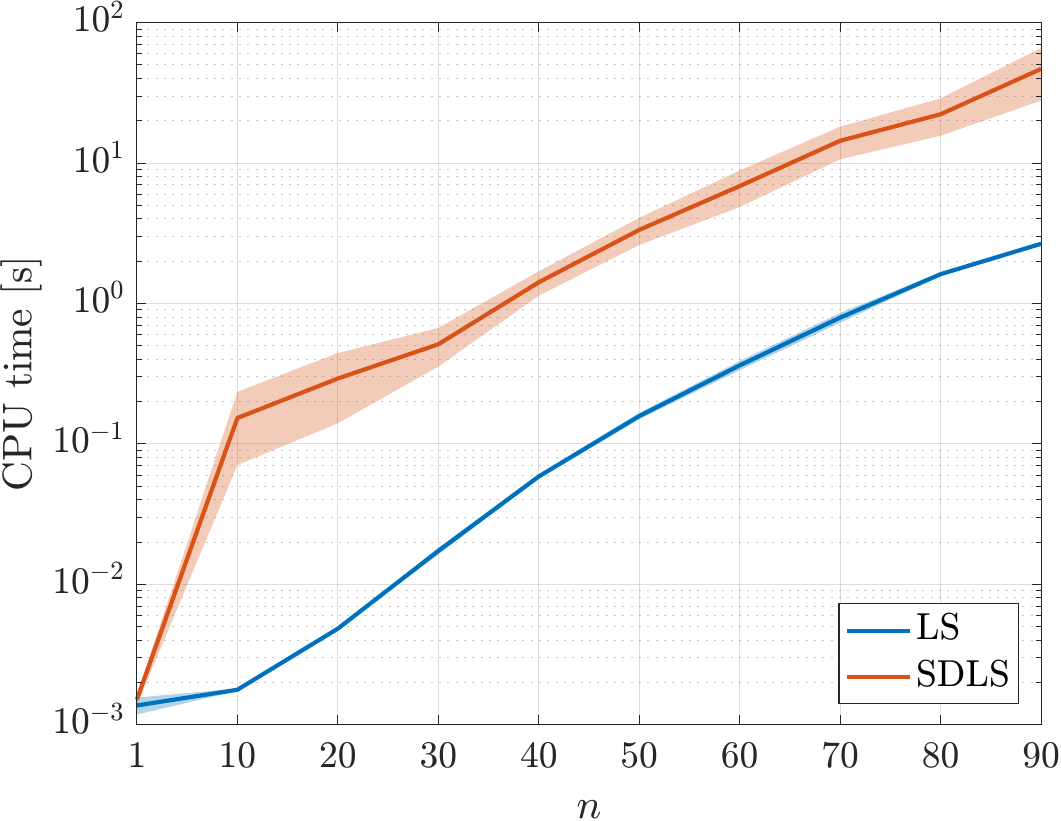}
	\caption{Computational time for solving the \gls{LS} in \eqref{eq:fitting} and related \gls{SDLS} variant, averaged over $20$ different numerical instances.}
	\label{fig:cputime}
\end{figure}

Next, we set $N=5000$ and compare the computational time required to solve the \gls{QP} in \eqref{eq:fitting}, and the related \gls{SDLS} variant including the \gls{SD} constraint $0\preccurlyeq\hat Q\preccurlyeq\lambda_{\max}(Q)$, for different values of $n$. Note that we actually solve (SD)LS from $3$ up to $8191$ decision variables. As illustrated in Fig.~\ref{fig:cputime}, solving the \gls{LS} problem only requires at least an order of magnitude less than solving the associated \gls{SDLS}.


\section{Conclusion}
By focusing on data-driven \gls{LS} problems with \gls{SD} constraints, we have derived probabilistic guarantees holding with high confidence on the spectrum of their optimal solutions once these constraints are removed. Our certificate allows one to solve a simple program, which frequently turns into a \gls{QP}, in place of the full \gls{SDLS} problem, while ensuring that the span of the resulting solution is $\varepsilon$-close to that enforced by the \gls{SD} constraints. Consistently, our bound shrinks as the number of data increases, it is distribution-free and only requires \gls{iid} samples.
As a consequence, when \gls{SDLS} is designed to learn an unknown function, we have shown how to bound the error between a gradient descent iterate minimizing the surrogate cost obtained with no SD constraints and the true minimizer.

Along the line of the results in \S \ref{sub:optimization}, future work will explore the link between these bounds and data-driven optimization. 
\bibliographystyle{ieeetr}
\bibliography{PaperCollection00.bib}

\end{document}